\documentclass[psamsfonts]{amsart}

\usepackage{amssymb,amsfonts}
\usepackage[all,arc]{xy}
\usepackage{enumerate}
\usepackage{mathrsfs}

\newtheorem{theorem}{Theorem}[section]
\newtheorem{claim}[theorem]{Claim}
\newtheorem{cor}[theorem]{Corollary}
\newtheorem{prop}[theorem]{Proposition}
\newtheorem{lem}[theorem]{Lemma}
\newtheorem{conj}{Conjecture}
\newtheorem*{conj*}{Conjecture}

\theoremstyle{definition}
\newtheorem{defn}[theorem]{Definition}

\newtheorem{exmp}[theorem]{Example}

\theoremstyle{remark}
\newtheorem{rem}[theorem]{Remark}

\makeatletter
\let\c@equation\c@theorem
\makeatother
\numberwithin{equation}{section}

\begin{document}

	\title{On Sensitivity of $k$-uniform Hypergraph Properties}

	\author{Stella Biderman, Kevin Cuddy, Ang Li, Min Jae Song}
	\date{September 25th, 2017}
	\begin{abstract}
		 In this paper we present a graph property with sensitivity $\Theta(\sqrt{n})$, where $n={v\choose2}$ is the number of variables, and generalize it to a $k$-uniform hypergraph property with sensitivity $\Theta(\sqrt{n})$, where $n={v\choose k}$ is again the number of variables. This yields the smallest sensitivity yet achieved for a $k$-uniform hypergraph property. We then show that, for even $k$, there is a $k$-uniform hypergraph property that demonstrates a quadratic gap between sensitivity and block sensitivity. This matches the largest known gap found by Ambainis and Sun (2011) for Boolean functions in general, and is the first known example of such a gap for a graph or hypergraph property.
	\end{abstract}
	\maketitle
	
	\tableofcontents
	
	\section{Introduction}
The sensitivity of a Boolean function $f: \{0,1\}^{n} \rightarrow \{0,1\}$ is the maximum over all inputs $x$ of the number of $i$ such that flipping the $i$-th bit of $x$ flips the value of $f(x)$. Block sensitivity, introduced by Nisan \cite{N}, is the maximum over all inputs $x$ of the maximum number of disjoint blocks of variables such that flipping all the bits in any single block flips the value of $f(x)$. We present these definitions with precise mathematical notation in section 2.  It follows from definition that for all $f$, $s(f) \leq bs(f)$, since we can just consider a partition of the input into blocks of size one.  Whether there is a polynomial bound in the other direction is a long-standing open question and is known as the Sensitivity Conjecture. In other words, do there exist $a,b\in\mathbb{R}$ such that $\forall f:\{0,1\}^n\to\{0,1\},  bs(f)<a\cdot s(f)^b$? A stronger conjecture by Nisan and Szegedy \cite{NS} is that $bs(f) \leq O(s(f)^{2})$. Currently, the best known upper bound on block sensitivity is exponential in sensitivity, as given by Ambainis et al. \cite{APV}; their upper bound is $bs(f) \leq \max(2^{s(f)-1}(s(f)-\frac{1}{3}),s(f))$. Recently, this bound has been improved further to $bs(f) \leq (\frac{8}{9}+o(1))s(f)2^{s(f)-1}$ by He, Li, and Sun \cite{HLS}.

The Sensitivity Conjecture motivates the search for Boolean functions with large gaps between sensitivity and block sensitivity. The largest known gap is quadratic and is due to Ambainis and Sun \cite{AS}. One of the earlier examples of a function demonstrating a quadratic gap was given by Rubinstein \cite{R}. Chakraborty \cite{C} explained how a minor modification to Rubinstein's function yields a cyclically invariant Boolean function also demonstrating a quadratic gap. This result raises the question: can more weakly symmetric functions, specifically graph properties and uniform hypergraph properties, demonstrate quadratic (or greater) gaps between sensitivity and block sensitivity? In this paper, we explore the sensitivity of $k$-uniform hypergraph properties and consider the following three conjectures of Laszlo Babai (all three conjectures are for every $k \in \mathbb{N}, k\geq 2$):

\begin{conj}
There exists a $k$-uniform hypergraph property $f$ such that $$s(f)=O(\sqrt{n})=O(v^{k/2}).$$
\end{conj}

\begin{conj}
There exists a $k$-uniform hypergraph property $f$ such that $s(f)$ and $bs(f)$ display a quadratic gap. That is,  $bs(f)=\Omega(s(f)^2)$.
\end{conj}

\begin{conj}
For every $k$-uniform hypergraph property $f,$ $$s(f)=\Omega(\sqrt{n})=\Omega(v^{k/2}).$$
\end{conj}

In Section 3, we answer Conjecture 1 in the affirmative by demonstrating, for every $k \in \mathbb{N}$, a $k$-uniform hypergraph property with sensitivity $s(f)=\Theta(\sqrt{n})=\Theta(v^{k/2})$. In Section 4, we partially answer Conjecture 2 in the affirmative by demonstrating a $k$-uniform hypergraph property with a quadratic gap between senstivity and block sensitivity for every even $k$.  However, for odd $k,$ the best we have accomplished is a gap of $\frac{2k}{k+1}$ in the exponent.  That is, for all odd $k$, we give a graph property $f$ such that 
$$bs(f)=\Theta(s(f)^{\frac{2k}{k+1}}).$$

Interestingly, in the case of $k=2$, we know that the function given in Section 4 is optimal in the sense of giving the largest gap between sensitivity and block sensitivity.  Turan \cite{T} has shown that for a graph property $f$ on $v$ vertices, $s(f) = \Omega(\sqrt{n})=\Omega(v)$. This also answers Conjecture 3 in the affirmative for the case of $k=2$, but for other values of $k$, the conjecture was open. After we raised this question in our initial preprint, it was answered negatively by Li and Sun \cite{LS}.

	\section{Preliminaries}
		We first introduce some complexity measures that can be used to analyze Boolean functions.

	\subsection{Definitions}
\begin{defn}\label{sensitivity}
For a Boolean function $f: \{0,1\}^n \to \{0,1\}$ and an input string $x \in \{0,1\}^{n}$, the \textbf{sensitivity of $\boldsymbol{f}$ at $\boldsymbol{x}$}, denoted as $\boldsymbol{s(f,x)},$ is $$s(f,x) =|\{i \in [n] \: | \: f(x) \neq f(x^{i})\}|$$ where 
$\lbrack n \rbrack$ is the set $\{1,2,\ldots ,n\}$ and $x^{i}$ is the string $x$ with the $i^{th}$ bit flipped, that is, changed either from 0 to 1 or from 1 to 0.  So the sensitivity of $f$ at $x$ is the number of single-bit changes to $x$ that change $f(x)$.  
\\
The \textbf{sensitivity of $\boldsymbol{f}$} or $\boldsymbol{s(f)}$ is then $$s(f) = \max\limits_{x \in \{0,1\}^{n}} s(f,x).$$
\end{defn}

There is a similar measure called block sensitivity, introduced by Nisan \cite{N}.  
\begin{defn}\label{blocksensitivity}
A \textbf{sensitive block for $\boldsymbol{f}$ at $\boldsymbol{x}$} is a subset $B \subseteq [n]$ such that 
$f(x) \neq f(x^{B})$, where $x^{B}$ is the string $x$ with all bits whose indices are in $B$ flipped.
For example, if $n=3, \: B_{1} = \{1,2\},$ and $x=101$, then $x^{B_{1}}$ is $011$.  
The \textbf{block sensitivity of $\boldsymbol{f}$ at $\boldsymbol{x}$} is then
$$bs(f,x) = \max_{B_{1}, \ldots , B_{k}} k$$ 
where the maximum is taken over all sets of disjoint sensitive blocks for $f$ at $x$.  This is the maximum number of disjoint sensitive blocks for $f$ at $x$.
The \textbf{block sensitivity of $\boldsymbol{f}$} or $\boldsymbol{bs(f)}$ is then defined to be
$$bs(f) = \max\limits_{x \in \{0,1\}^{n}} bs(f,x).$$
\end{defn}

Clearly, for any Boolean function $f$, we have 
		\begin{equation}
			s(f) \leq bs(f). \nonumber
		\end{equation}

\begin{defn}\label{oneandzero}
It will be useful when presenting proofs to use the notation $\boldsymbol{s^{0}(f)}$ for the sensitivity when changing $f$'s value from 0 to 1 
and $\boldsymbol{s^{1}(f)}$ for the sensitivity when changing 
the value from 1 to 0.  Sensitivity is then $\max \{ s^{0}(f), s^{1}(f)\}$.  The same notation will be used for block sensitivity when required.  
\end{defn}

\begin{defn}
Two complexity measures on Boolean functions $A(f)$, $B(f)$ are called \textbf{polynomially related} if there exist $n$, $m$, $s$, $t \in \mathbb{N}$ such that 
for any Boolean function $f$, $s \cdot A(f)^{n} \leq B(f)$ and $t \cdot B(f)^{m} \leq A(f)$.  
\end{defn}

As stated before, $s(f) \leq bs(f)$, but a polynomial bound in the other direction is unknown.  

		\begin{defn}
	A Boolean function $f: \{0,1\}^{v \choose{2}} \to \{0,1\}$ is called a \textbf{graph property} if for every input $x = (x_{(1,2)},...,x_{(n-1,n)})$ and every permutation $\sigma \in S_v$, we have
			\begin{equation}
				f(x_{(1,2)},x_{(1,3)},...,x_{(n-1,n)}) = f(x_{(\sigma (1), \sigma (2))},...,x_{(\sigma (n-1), \sigma (n))}). \nonumber
			\end{equation}
		\end{defn}
Here the Boolean string $x$ is interpreted as a graph in the following manner: $x_{(1,2)}$ is 1 if there is an edge connecting vertex 1 and vertex 2, and it is 0 if there is no such edge.  
The labelling of the vertices is of course arbitrary; thus we could have two distinct Boolean strings representing the ``same'' graph (two isomorphic graphs).  The length of the input string 
must be ${v \choose{2}}$ for some $v \in \mathbb{N}$ because this is how many edge relations must be specified to define a graph on $n$ vertices.  A graph property $f$ thus satisfies $f(G) = f(H)$ whenever $G$ and $H$ are isomorphic graphs.	Similarly, we can define $k$-uniform hypergraph property.
			\begin{defn}
				A Boolean function $f: \{0,1\}^{v \choose{k}} \to \{0,1\}$ is called a \textbf{$k$-uniform hypergraph property} if for every input $x = (x_{(1,2,...,k)},...,x_{(n-k+1,...,n-1,n)})$ and every permutation $\sigma \in S_v$, we have
			\begin{equation}
				f(x_{(1,2,...,k)},...,x_{(n-k+1...,n-1,n)}) = f(x_{(\sigma (1), \sigma (2),..., \sigma(k))},...,x_{(\sigma (n-k+1),...,\sigma (n-1), \sigma (n))}). \nonumber
			\end{equation}
			\end{defn}
Again, this means that $f$ takes on the same value for isomorphic hypergraphs.

	\subsection{Previous Results}

		The largest known gap between sensitivity and block sensitivity is quadratic. A classic example of a function demonstrating quadratic separation is Rubinstein's function \cite{R}.
	\begin{exmp} ({\it Rubinstein's function})
		Let $n = k^2$ for some even $k$. We partition the $n$ variables into $k$ consecutive blocks of $k$ variables. We denote the blocks by $B_1,...,B_k$. Let $f(x) = 1$ if and only if there exists a block $B_i$ such that exactly two consecutive bits $x_{(i-1)k+j}, x_{(i-1)k+(j+1)} \in B_i$ take on the value 1 and all the other bits in $B_i$ are 0. The sensitivity of $f$ is $2k$ and the block sensitivity is $k^2/2$.
	\end{exmp}

	Chakraborty modified Rubinstein's function and constructed a cyclically invariant Boolean function with a quadratic gap between sensitivity and block sensitivity.
\begin{defn} Let $x = x_1x_2...x_n$ where each $x_i \in \{0,1\}$. Then for $0 < l < n$, the binary string $x_{l+1}x_{l+2}...x_nx_1...x_l$ is a \textbf{cyclic shift} of $x$.
\end{defn}

	\begin{exmp}
		Let $g : \{0,1\}^{k^2} \rightarrow \{0,1\}$ be Rubinstein's function. Define $f : \{0,1\}^{k^2} \rightarrow \{0,1\}$ as $f(x) = 1$ if and only if $g(x') = 1$ for some $x'$ which is a cyclic shift of $x$. $f$ has sensitivity $2k$ and block sensitivity $k^2/2$.
	\end{exmp}
	Cook et al. \cite{CDR} proved that for a CREW PRAM (Concurrent Read Exclusive Write Parallel RAM - a theoretical construct simulating a type of computer), it takes at least $s(f)$ steps to compute the value of $f$. Nisan \cite{N} later introduced the concept of block sensitivity and observed that $s(f)$ can be replaced by $bs(f)$ for the lower bound of CREW PRAM's time complexity.  Moreover, he showed that the time complexity of CREW PRAM is equal to (up to a constant factor) the logarithm of block sensitivity.

	Nisan and Szegedy \cite{NS} proved that block sensitivity is polynomially related to other complexity measures such as decision tree complexity, certificate complexity, and degree of the polynomial representation of a Boolean function.  For more information on such polynomial relations between complexity measures, we refer the reader to \cite{HKP}.  The surprising polynomial relations between almost all types of complexity measures of Boolean functions motivate the search for an answer to the Sensitivity Conjecture. 
	
	H.-U. Simon \cite{S} proved that for any Boolean function $f$, we have $s(f) \geq 1/2\log n -1/2\log{\log n} + 1/2$, where $n$ is the number of effective variables, that is, the number of variables on which the function actually depends. This lower bound is tight up to the additive  $O(\log{\log n})$. This gives us an exponential upper bound on the block sensitivity in terms of sensitivity. Kenyon and Kutin \cite{KK} introduced the notion of $\ell$-block sensitivity and gave a better upper bound, $bs(f) \leq O(e^{s(f)}\sqrt{s(f)})$, which is still exponential in sensitivity. Ambainis et al. \cite{APV} improved this bound to $bs(f) \leq 2^{s(f)-1}(s(f)-\frac{1}{3})$, and recently, it was improved further to $bs(f) \leq (\frac{8}{9}+o(1))s(f)2^{s(f)-1}$ by He, Li, and Sun \cite{HLS}.

	A function $f$ is called monotone if $f(x) \leq f(y)$ whenever $y$ can be obtained by flipping the 0's of $x$ to 1's. Nisan \cite{N} proved that for monotone Boolean functions, sensitivity and block sensitivity are asymptotically the same.  

	For graph properties, Turan \cite{T} proved that if $f$ is a graph property on $v$ vertices, then $s(f) = \Omega(v) = \Omega(\sqrt{n})$, which implies that the gap is at most quadratic for graph properties. In the same paper, he provides an example, the ``isolated vertex'' property, which shows that the lower bound is tight.

	\begin{exmp} ({\it Isolated vertex property}) 
		Let $f$ be the graph property on $v$ vertices such that $f(G) = 1$ if and only if $G$ has an isolated vertex. The sensitivity of $f$ is $v-1$. Note that this property is monotone, so $bs(f) = s(f)$.
	\end{exmp}
			
\section{$k$-Uniform  Hypergraph Property with $\Theta(v^{k/2})$ Sensitivity}
	The isolated vertex property in the previous section gives a graph property with sensitivity exactly $\Theta(v)$. In this section, we show $k$-uniform graph properties with $\Theta(v^{k/2})$ sensitivity for all $k \geq 3$.

	\begin{theorem}
		Given any $k,i \in \mathbb{N}$ with $k\geq 3$, $1 \leq i < k$ and any $t \in [0,1]$, there exists a $k$-uniform hypergraph property $f$ such that $s(f) = \max\{O(v^{i(1-t)}), \Theta(v^{k-i(1-t)})\}$.
	\end{theorem}
	\begin{proof}
		Let $\mathcal{H}$ be a $k$-uniform clique on $\Theta(v^t)$ vertices.  Define a $k$-uniform hypergraph property $f$ such that $f=1$ if there exists a copy of $\mathcal{H}$ inside the graph such that given any edge $e$ not entirely in $\mathcal{H}$, we have
		\begin{equation}
			|V(\mathcal{H}) \cap e| < i. \nonumber
		\end{equation}
		\indent Let us refer to this condition as the ``isolation condition". We claim that this is the desired $k$-uniform hypergraph property. First we consider $s^0(f)$. If there is no  $\mathcal{H}$ satisfying the condition inside the graph, we call $v^t$ vertices $\{w_1,...w_{v^t}\}$ a sensitive tuple if we can form $\mathcal{H}$ satisfying the isolation condition by either removing or adding one edge. 
		\begin{lem}
			Every sensitive tuple contains exactly one sensitive edge.
		\end{lem}
		\begin{proof}
			If the tuple can form $\mathcal{H}$ by removing one edge, then it is a desired clique but does not fit the isolation condition. There can only be one edge that needs to be removed. If there were two edges that violated the isolation condition, then removing only one of them wouldn't change the value of the function, and hence, the tuple wouldn't be a sensitive tuple.
If the tuple can form a desired $\mathcal{H}$ by adding one edge, then it satisfied the isolation condition and is one edge short of being a clique. But then there's only one edge that can be added to those points, as every other edge in the clique is already present.
		\end{proof}
		 Then every sensitive tuple contains exactly one sensitive edge. Let $\mathcal{F}$ denote the set of all sensitive tuples.  We have
		\begin{equation}
			s^0(f) \leq |\mathcal{F}|. \nonumber
		\end{equation}
		Then we can get an upper bound on $s^0(f)$ by finding an upper bound on $|\mathcal{F}|$. We claim that no two elements of $\mathcal{F}$ can share more than $i-1$ vertices. Suppose we have $A=\{w_1,...,w_{j}\}$ and $B=\{w_1^\prime,...,w_{j}^\prime\}$ two distinct sensitive tuples such that they have $i$ common vertices $\{v_1,...,v_i\}$ and $i<j$ (here $j$ is $\Theta(v^t)$).  Since $j$ grows with $v$ and $k$ is constant, we can consider larger graphs with $v^{t} > k + 2$.  Using this assumption, if both of our tuples can form a desired $\mathcal{H}$ by adding one edge, they are both almost cliques - they each have every edge required except one.  Then there is already at least one edge $e_{1}$ such that $\{v_1,..,v_i\} \subset e_{1} \subset A$ and a different edge $e_{2}$ satisfying $\{v_1,..,v_i\} \subset e_{2} \subset B$, 
		because in both $A$ and $B$ every subset of size $k$ is present as an edge except one, and there are more than one such subsets containing $\{v_1,...,v_i\}$ because 
		$j > k+2 > i+2$.  
		Then adding one edge to complete $A$ would still leave $e_{2}$, which is not contained in $A$ and violates our isolation condition, so $A$ would not be a 
		desired isolated subgraph.  Similarly adding an edge to $B$ cannot produce an isolated clique because of $e_{1}$. \\
		\indent
		If one of these two tuples can form a desired $\mathcal{H}$ by adding one edge and the other by removing an edge, then either $A$ or $B$ is already a clique with one extra edge violating the isolation condition, so similarly any attempt to make $A$ or $B$ into the desired $\mathcal{H}$ will not be possible because of overlapping edges containing points in
		$\{v_1,...,v_i\}$, which are required by the definition of clique.  \\
		 \indent
		If both tuples can form $\mathcal{H}$ by removing one edge, again there are at least $j-i \choose{k-1}$ edges containing vertices in $\{v_1,...,v_i\}$ that are overlapping 
		into our desired cliques, and with enough vertices this is far too many.  \\
		 \indent Thus we know that any set of $i$ vertices is contained in no more than one sensitive tuple in $\mathcal{F}$ and any sensitive tuple contains $v^t \choose{i}$ distinct sets of $i$ vertices. So we have
		\begin{equation}
			s^0(f) \leq |\mathcal{F}| \leq \frac{{v \choose{i}}}{{v^t \choose{i}}} = O(v^{i(1-t)}). \nonumber
		\end{equation}
		Finally we consider $s^1(f)$. When there exists such an $\mathcal{H}$ inside the graph, the only way to eliminate it is to either remove an edge inside $\mathcal{H}$ or add an edge with more than $i-1$ vertices in $V(\mathcal{H})$. Thus we have
		\begin{equation}
			s^1(f) = {v^t \choose{k}} + {v^t \choose{i}}{v-i \choose{k-i}} = \Theta (v^{k-i(1-t)}). \nonumber
		\end{equation}
		The second term represents the case of adding an edge with at least $i$ vertices in $V(\mathcal{H})$. We only need to consider the second term because it is the asymptotically dominating term. It then follows that we have
		\begin{equation}
			s(f)= \max\{s^0(f),s^1(f)\} = \max\{ O(v^{i(1-t)}),  \Theta(v^{k-i(1-t)})\}. \nonumber
		\end{equation}
	\end{proof}
	By looking at special cases of this result, we can find a $k$-uniform hypergraph property with sensitivity $O(v^{k/2})$.
	\begin{cor}
		There exists a $k$-uniform hypergraph property with sensitivity $O(v^{k/2})$, for all $k$.
	\end{cor}
	\begin{proof}
		An example with $k=2$ was given before. For $k\geq 3$, any solution to the equation $i(1-t)=\frac{k}{2}$ for $t\in[0,1]$ and $i\in\mathbb{N}$ works, and specifically for every $i\geq\frac{k}{2}$, setting $t=1-\frac{k}{2i}$ yields the desired result.
	\end{proof}

	We now introduce a lemma which will allow us, in certain cases, to obtain a tight lower bound on $s^0(f)$.

	\begin{lem}
		Let $q$ be a prime power and $d \leq q$. $\forall\ell\in\mathbb{N}$, there exists a collection of sets $S_1, ..., S_m \subseteq [q^{\ell+1}]$ such that $|S_i| = q$ for all $i$ and $|S_i \cap S_j| < d$ for $i \neq j$ and $m = q^{d\ell}$.
	\end{lem}
	\begin{proof}
		Let $f : \mathbb{F}_q \rightarrow \mathbb{F}^\ell_q$ be such that $f(x) = (f_1(x), ... ,f_k(x))$ where each $f_i$ is a degree $k-1$ polynomial over $\mathbb{F}_q$. Then each $f$ corresponds to a set of $q$ $(k+1)$-tuples, $S_f = \{(x,f_1(x),...,f_k(x))\text{ }  |\text{ } x \in \mathbb{F}_q\}$. 

		If $g \neq f$, then the sets $S_f$ and $S_g$ intersect in at most $k-1$ points since the equation $f_1(x) = g_1(x)$ already has at most $k-1$ solutions in $\mathbb{F}_q$. We can relabel $S_f$ by associating $i \in [q^{k+1}]$ to each $f : \mathbb{F}_q \rightarrow \mathbb{F}^k_q$. There are $q^k$ distinct polynomials over $\mathbb{F}_q$ of degree $k-1$, so there are $(q^k)^k$ distinct sets of $(k+1)$-tuples that satisfy the above property. Hence, we can construct a collection of sets $S_1, ..., S_{q^{k^2}}$ such that $|S_i| = q$ for all $i$ and $|S_i \cap S_j| < k$ for $i \neq j$.
	\end{proof}
	\begin{prop}
		The upper bound in Corollary 3.3 is tight. That is, $\forall k \in \mathbb{N}, \exists i\in\mathbb{N},t\in[0,1]$ such that if $f$ is the function in Corollary 3.3 with parameters
	$t$, $i$, and $k$, then $s(f)=\Theta(v^\frac{k}{2})$
	\end{prop}
	\begin{proof}
		It is sufficient to prove that, in a special case of Theorem 3.1, the upper bound on $s^0(f)$ is tight, as we have already shown $s^{1}(f)$ is tight. We will prove this by evoking the lemma just introduced.

		First, consider the case of $k$ even. $i=\frac{k}{2},t=1$ is an assignment yielding $s(f)=O(v^\frac{k}{2})$, which we will prove to be tight.  Let $q$ be a prime power such that $k+1<q<2(k+1)$, the existence of which is guaranteed by Bertrand's Postulate. Let $\ell=\lfloor\log_q(v)-1\rfloor\geq\log_q(v)-2$, and $d=\frac{k}{2}$. Then by the lemma we have $m=q^{d\ell}=\Omega(v^\frac{k}{2}q^{-2\frac{k}{2}})=\Omega(v^\frac{k}{2}q^{-k})=\Omega(v^\frac{k}{2})$ many sets $S_1,\ldots S_m$ such that $|S_i|\geq k+1$ and $\forall i\neq j , |S_i\cap S_j|<\frac{k}{2}$. Let $S'_i$ be any subset of $S_i$ such that $|S'_i|=k+1$. Note that $\forall i\neq j,|S'_i\cap S'_j|<\frac{k}{2}$.

		Now we can construct an input for which $s^0(f)=\Omega(v^\frac{k}{2})$. Label the vertices $1,\ldots ,v$ and partition a subset of these into the sets $S'_i$. Make the sets of points $S'_i$ be cliques, and then remove any one edge from each of the cliques. Since no two share $\frac{k}{2}$ points, the isolation condition is automatically satisfied, so re-adding the removed edge to any of the almost-cliques changes the value of $f$ from $0$ to $1$. There are $\Omega(v^\frac{k}{2})$ many sets $S'_i$, so $s^0(f)=\Omega(v^\frac{k}{2})$.

		Next, for the case of $k$ odd.  $i=\frac{k+1}{2},t=\frac{1}{k+1}$ is an assignment yielding\\
$s(f)=O(v^\frac{k}{2})$, which we will prove to be tight. Let $q$ be a prime power such that $\frac{1}{2}v^t<q<v^t$, the existence of which is again guarenteed by Bertrand's Postulate.  Let $\ell=\frac{1}{t}-1=k$, and $d=\frac{k+1}{2}$. Then by the previous lemma, we have that $m=q^{d\ell}=\Omega(v^{td\ell})=\Omega(v^{\frac{1}{k+1}\frac{k+1}{2}k})=\Omega(v^\frac{k}{2})$ many sets $S_1,\ldots S_m$ such that $|S_i|\geq k+1$ such that $\forall i\neq j,|S_i\cap S_j|<\frac{k}{2}$. We proceed to construct our input as before, and again obtain $s^0(f)=\Omega(v^\frac{k}{2})$.
	\end{proof}

	\begin{rem}
		Although it might appear that Lemma 3.4 could be used to obtain a general tight lower bound on $s^0(f)$, there is an important caveat. Since we consider $\mathbb{F}^\ell_q$, $\ell$ must be a positive integer. However, using our argument, except in the case of $t=0$, $l$ is an integer if and only if $\frac{1}{t}$ is, and since $t\in[0,1]$, this is not guaranteed in general.
	\end{rem}

\section{Gaps in Sensitivity and Block Sensitivity for Graph Properties}
	In this section, we give a $k$-uniform hypergraph property that demonstrates a quadratic gap between sensitivity and block sensitivity for every even $k$. We first start with a graph property.

	\begin{theorem}
		There exists a graph property $f$ such that $s(f) = \Theta(v)$ and $bs(f) = \Theta(v^2)$.
	\end{theorem} 

	\begin{proof}
		Let $f$ be the graph property, ``there exists an isolated
		 triangle.''  By isolated triangle, we mean three vertices $a,b,c$ with edges $\{a,b\}, \{a,c\}, \{b,c\}$ and no 
		edges between any one of $a$,$b$, or $c$ and any $v \in V\setminus\{a,b,c\}$.

		Consider the block sensitivity on an empty graph input. There are at least $ \frac{{v \choose 3}}{3(v-3)} = \Omega(v^2)$ edge-disjoint triangles on $v$ vertices. Hence, $bs(f) = \Omega(v^2)$ because we can just group the edges as these disjoint triangles. Since $bs(f) \leq {v \choose 2}$ because there are only this many vertices, it follows that $bs(f) = \Theta(v^2)$.

		For sensitivity, we consider $s^0(f)$ first. To change from 0 to 1, we need to either add an edge to complete an isolated triangle or remove an edge so that a triangle becomes isolated. We call a 3-tuple $(v_1,v_2,v_3)$ sensitive if adding or removing an edge from the graph will make the $(v_1,v_2,v_3)$ vertices an isolated triangle.

	\begin{claim}
		Sensitive 3-tuples are pairwise vertex disjoint. 
	\end{claim}
	\begin{proof}
		To see this, let two distinct sensitive 3-tuples be $(a,v_1,v_2)$ and $(a,w_1,w_2)$. Since $(a,v_1,v_2)$ is 1 edge away from being an isolated triangle, $a$ cannot share an edge with both $w_1$ and $w_2$. In fact, $a$ shares an edge with exactly one of them because otherwise $(a,w_1,w_2)$ is not 1 edge from being an isolated triangle. By the same logic, $a$ shares an edge with exactly one of $v_1, v_2$. This implies that induced subgraphs on $(a,v_1,v_2)$ and $(a,w_1,w_2)$ are trees. However, this means that $(a,v_1,v_2)$ cannot be sensitive since there is an edge connecting $a$ to a vertex other than $v_1,v_2$. Hence, sensitive 3-tuples cannot share a vertex. 
	\end{proof}

	Now we can prove that $s^0(f) = \Theta(v)$. Let $C_1, C_2, ... C_m$ be sensitive 3-tuples on a graph $G$. Since two sensitive 3-tuples cannot share a vertex, a vertex uniquely determines a sensitive 3-tuple. Hence, $m \leq v$. Since we can have $v/3$ vertex disjoint trees consisting of 3 vertices on $G$, $s^0(f) = \Theta(v)$.

	Next we consider $s^1(f)$. To change from 1 to 0, we need to either remove an edge from the isolated triangle or add an edge to it so that it is no longer isolated. For the first case, it is easy to see that 3 bits are sensitive. For the second case, at most $3(v-3)$ bits can make the triangle not isolated, so $s^1(f) = \Theta(v),$ and therefore $s(f) = \max\{\Theta(v),\Theta(v)\}=\Theta(v)$.
	\end{proof}

	Then we can generalize this graph property to a $k$-uniform hypergraph property for any $k \in \mathbb{N}$, which gives a quadratic gap for $k$ even and a nearly quadratic gap for $k$ odd.
	\begin{theorem}
		Let $k,i \in \mathbb{N}$ and $1 \leq i \leq k/2$.  There exists a $k$-uniform hypergraph property $f$ on $v$ vertices such that $s(f) = \Theta(v^{k-i})$ and $bs(f) = \Theta(v^{k})$.
	\end{theorem}
	\begin{proof}
		Consider the function from Theorem 3.1, with $t$ set to 0. Since $i\leq\frac{k}{2}$, we have that $i\leq k-i$, so $s(f)=\Theta(v^{k-i})$. We now wish to show that $bs(f)=\Theta(v^k)$.
		\begin{lem}
			In a $k$-uniform hypergraph on $v$ vertices, there are $\Omega (v^k)$ edge disjoint $K_{k+1}^{(k)}$.
		\end{lem}
		\begin{proof}
			Note that $t=0$ because we are considering cliques of a constant size, $k+1$.  There are $v \choose{k+1}$  many ($k+1$)-cliques, and for every clique chosen, we eliminate any other cliques with more than $k-1$ points in common. In this way we get at least
		\begin{equation}
			\frac{{v \choose{k+1}}}{{k+1 \choose{k}}(v-k)} = c \cdot v^{k} \nonumber
		\end{equation}		 
		many cliques, and any two of them have no more than $k-1$ points in common, which guarantees that they are edge-disjoint cliques (no edge is in two distinct cliques).
		\end{proof}
		Consider the empty graph. By the lemma, there are $\Omega (v^k)$ disjoint cliques and each of them is a sensitive block. Hence we have $bs(f) = \Theta (v^k)$ by applying the trivial upper bound on block sensitivity.
	\end{proof}
	
	\begin{cor}
		For all $k \in \mathbb{N}$, there exists a $k$-uniform hypergraph property on $v$ vertices such that $s(f) = \Theta(v^{\lceil k/2 \rceil})$ and $bs(f)= \Theta (v^k)$.
	\end{cor}
	\begin{proof}
		When $k$ is even, choose $i=k/2$ and when $k$ is odd, choose $i=(k-1)/2$ in above proposition.
	\end{proof}
	
	\begin{rem}
		This proof can be generalized past the specific subgraph of  $K_{k+1}^{(k)}$, but to what extent, we are not certain.  One possible type of subgraphs for which the proof could apply are ``near-cliques,'' or subgraphs that are a constant number of edges away from being a clique. The idea is to preserve the fact that there is only one way to complete a sensitive tuple by adding a single edge.  Certain types of subgraphs do not fit this criterion. For example, a simple cycle of ordered vertices $\{v_1, \ldots , v_w\}$ with one ``chord'' edge added, say $\{v_1,v_{w/2}\}$, can be completed (up to isomorphism) from a simple cycle in a number of ways that is $\Theta(w)$, where the cycle has $w$ vertices.  
	\end{rem}

	\section{Further Work and Open Problems}
		We have shown that there exists a $k$-uniform hypergraph property $f$ such that $s(f) \leq O(v^{k/2})$, and that for $k$ even, there is a $k$-uniform hypergraph property that demonstrates quadratic separation between $bs(f)$ and $s(f)$. Hence, our results answer Conjecture 1 in the affirmative and provide partial answer to Conjecture 2. It remains unknown whether Conjecture 2 holds for odd $k$. Conjecture 3, on the other hand, has been answered negatively by Li and Sun \cite{LS} after our initial preprint.

		A possible continuation of this work appears to be the following modification to the function presented in this paper: the function in Theorem 3.1 can be viewed as an indicator function of a subgraph with certain degree of isolation. Although in this paper we restrict our attention to the case when the subgraph is a clique, there are other subgraphs that we can look for to obtain the same result. It is possible that choosing a different subgraph might yield better results, especially in the case of $k$ odd.
	
	In summary, we still have the following statement open:
	\begin{conj*}
		There exists, for every $k$ odd, a $k$-uniform hypergraph property $f$ such that $s(f)$ and $bs(f)$ display a quadratic gap. That is,  $bs(f)=\Omega(s(f)^2)$.
	\end{conj*}

	\subsection*{Acknowledgements}
		We would like to thank our mentor, Professor Laci Babai, for introducing us to the Sensitivity Conjecture and its related problems, and providing us with many helpful comments.  Without his guidance and suggestions, this paper would not have been possible.  We also thank him and Professor Stuart Kurtz for organizing the wonderful Computer Science REU, during which this paper was completed.

\end{document}